\pdfoutput=1
\documentclass[sigconf,authorversion]{acmart}

\usepackage[english]{babel}  
\usepackage{amssymb}
\usepackage{arydshln}  

\usepackage{xparse}

\usepackage[shortlabels]{enumitem}
\usepackage{boxedminipage}

\usepackage[capitalise]{cleveref}

\makeatletter
\def\renewtheorem#1{%
  \expandafter\let\csname#1\endcsname\relax
  \expandafter\let\csname c@#1\endcsname\relax
  \gdef\renewtheorem@envname{#1}
  \renewtheorem@secpar}
\def\renewtheorem@secpar{\@ifnextchar[{\renewtheorem@numberedlike}{\renewtheorem@nonumberedlike}}
\def\renewtheorem@numberedlike[#1]#2{\newtheorem{\renewtheorem@envname}[#1]{#2}}
\def\renewtheorem@nonumberedlike#1{  
\def\renewtheorem@caption{#1}
\edef\renewtheorem@nowithin{\noexpand\newtheorem{\renewtheorem@envname}{\renewtheorem@caption}}
\renewtheorem@thirdpar
}
\def\renewtheorem@thirdpar{\@ifnextchar[{\renewtheorem@within}{\renewtheorem@nowithin}}
\def\renewtheorem@within[#1]{\renewtheorem@nowithin[#1]}

\newtheoremstyle{framedthmenv}%
  {0cm}
  {0cm}
  {\@acmdefinitionbodyfont}
  {\@acmdefinitionindent}
  {\@acmdefinitionheadfont}
  {:}
  {.5em}
  {\thmname{#1}\thmnumber{ #2}\thmnote{ {\@acmdefinitionnotefont(#3)}}}
\makeatother

\theoremstyle{acmplain}
\renewtheorem{theorem}{Theorem}[section]
\renewtheorem{proposition}[theorem]{Proposition}
\renewtheorem{lemma}[theorem]{Lemma}
\renewtheorem{corollary}[theorem]{Corollary}
\theoremstyle{acmdefinition}
\renewtheorem{definition}[theorem]{Definition}

\theoremstyle{framedthmenv}
\newtheorem{problem}{Problem} 
\crefname{problem}{Problem}{Problems}
\Crefname{problem}{Problem}{Problems}
\newtheorem{algorithm}{Algorithm} 
\crefname{algorithm}{Algorithm}{Algorithms}
\Crefname{algorithm}{Algorithm}{Algorithms}
\newtheorem{remark}{Remark}
\crefname{remark}{Remark}{Remarks}
\Crefname{remark}{Remark}{Remarks}

\theoremstyle{acmplain}

\usepackage{tikz}
\usetikzlibrary{arrows,positioning} 
\tikzset{
    >=stealth',
}

\newcommand{\storeArg}{} 

\newcommand{\bigO}[1]{O(#1)} 
\newcommand{\softO}[1]{\mathchoice{\tilde{O}\left(#1\right)}{O\tilde{~}(#1)}{O\tilde{~}(#1)}{O\tilde{~}(#1)}} 
\newcommand{\expmatmul}{\omega} 
\newcommand{\algoname}[1]{{\normalfont\textsc{#1}}} 
\newcommand{\problemname}[1]{{\normalfont\textsc{#1}}} 
\newcommand{\algoword}[1]{\emph{\textsf{#1}}} 
\newcommand{\assign}{\leftarrow} 
\newcommand{\ZZ}{\mathbb{Z}} 
\newcommand{\NN}{\mathbb{Z}_{\ge 0}} 
\newcommand{\ZZp}{\mathbb{Z}_{> 0}} 
\newcommand{\tuple}[1]{\boldsymbol{#1}}  
\newcommand{\sumVec}[1]{|#1|} 

\newcommand{\subTuple}[2]{{#1}_{#2}} 
\newcommand{\var}{x} 
\newcommand{\field}{\mathbb{K}} 
\newcommand{\polRing}{\field[\var]} 
\newcommand{\rdim}{m} 
\newcommand{\cdim}{n} 
\newcommand{\matrk}{r} 
\newcommand{\matSpace}[1][\rdim]{\renewcommand\storeArg{#1}\matSpaceAux} 
\newcommand{\polMatSpace}[1][\rdim]{\renewcommand\storeArg{#1}\polMatSpaceAux} 
\newcommand{\matSpaceAux}[1][\storeArg]{\field^{\storeArg \times #1}} 
\newcommand{\polMatSpaceAux}[1][\storeArg]{\polRing^{\storeArg \times #1}} 
\newcommand{\row}[1]{\mathbf{\MakeLowercase{#1}}} 
\newcommand{\rowgrk}[1]{\boldsymbol{#1}} 
\newcommand{\mat}[1]{\mathbf{\MakeUppercase{#1}}} 
\newcommand{\matt}[1]{\mathbf{\hat{\MakeUppercase{#1}}}} 
\newcommand{\matrow}[2]{{#1}_{#2,*}} 
\newcommand{\matrows}[2]{{#1}_{#2,*}} 
\newcommand{\matcol}[2]{{#1}_{*,#2}} 
\newcommand{\matcols}[2]{{#1}_{*,#2}} 
\newcommand{\trsp}[1]{#1^\mathsf{T}} 
\newcommand{\diag}[1]{\mathrm{diag}(#1)}  
\newcommand{\idMat}[1][\rdim]{\mat{I}_{#1}} 
\newcommand{\matz}{\mat{0}}  

\newcommand{\rdeg}[2][]{\mathrm{rdeg}_{{#1}}(#2)} 
\newcommand{\cdeg}[2][]{\mathrm{cdeg}_{{#1}}(#2)} 
\newcommand{\leadingMat}[2][]{\mathrm{lm}_{#1}(#2)} 
\newcommand{\leadingMatBig}[2][\unishift]{\mathrm{lm}_{#1}\left(#2\right)} 
\newcommand{\shiftMat}[1]{\mathbf{\var}^{#1}} 
\newcommand{\shiftSpace}[1][\cdim]{\ZZ^{#1}} 
\newcommand{\unishift}{\mathbf{0}} 
\newcommand{\shift}[2][s]{#1_{#2}} 
\newcommand{\shifts}[1][s]{\tuple{#1}} 
\newcommand{\shiftss}[1][s]{\tuple{\hat{#1}}} 
\newcommand{\shiftt}{\shifts[t]} 
\newcommand{\amp}[1][\shifts]{\operatorname{amp}(#1)} 
\newcommand{\expand}[1]{\overline{#1}} 
\newcommand{\minDeg}{\delta} 
\newcommand{\minDegs}{\boldsymbol{\minDeg}} 
\newcommand{\degBd}{\delta} 
\newcommand{\piv}{\pi} 
\DeclareDocumentCommand{\pivSupp}{ O{\shifts} O{\pmat} }{\boldsymbol{\piv}_{#1}(#2)} 
\newcommand{\app}{\row{p}} 
\newcommand{\appbas}{\matt{P}} 
\newcommand{\order}{d} 
\newcommand{\orders}{\tuple{\order}} 
\newcommand{\modApp}[2]{\operatorname{\mathcal{A}}_{#1}(#2)} 
\newcommand{\xDiag}[1]{\mathbf{\var}^{#1}} 
\newcommand{\popov}{\mat{P}} 
\newcommand{\reduced}{\mat{R}} 
\newcommand{\hermite}{\mat{H}} 
\newcommand{\pmat}{\mat{M}} 
\newcommand{\pmatt}{\matt{M}} 
\newcommand{\pmatSpace}{\polMatSpace[\rdim][\cdim]} 
\newcommand{\umat}{\mat{U}} 
\newcommand{\bmat}{\mat{B}} 
\newcommand{\cmat}{\mat{C}} 
\newcommand{\kerbas}{\mat{K}} 
\newcommand{\colbas}{\mat{C}} 
\newcommand{\rfac}{\mat{S}} 

\newlist{algosteps}{enumerate}{3}
\crefname{algostepsi}{Step}{Steps}
\crefname{algostepsi}{Step}{Steps}
\Crefname{algostepsii}{Step}{Steps}
\Crefname{algostepsii}{Step}{Steps}
\Crefname{algostepsiii}{Step}{Steps}
\Crefname{algostepsiii}{Step}{Steps}
\newenvironment{algobox}{
  \newcommand{\algoInfo}[1]{
    \begin{algorithm}
    \emph{\algoname{##1}}
  }
  \newcommand{\dataInfos}[2]{
    \algoword{##1:}
      \begin{itemize}[leftmargin=0.8cm]
          ##2
      \end{itemize}}
  \newcommand{\dataInfo}[2]{
    \algoword{##1:} ##2 }
  \newcommand{\algoSteps}[1]{
    \setlist[algosteps,1]{leftmargin=0.5cm}
    \setlist[algosteps,2]{leftmargin=0.4cm}
    \setlist[algosteps,3]{leftmargin=0.4cm}
    \begin{algosteps}[label=\textbf{\arabic*.},ref=\arabic*]
        ##1
    \end{algosteps}
  }
  \begin{figure}[ht]
  \centering
  \addtolength\fboxsep{0.1cm}
  \begin{boxedminipage}{0.99\columnwidth}
  }
  {
  \end{algorithm}
  \end{boxedminipage}
  \end{figure}
}

\begin{document}

\fancyhead{} 
\title{Computing Popov and Hermite Forms \\ of Rectangular Polynomial Matrices}

\author{Vincent Neiger}
\authornote{Part of the research leading to this work was conducted while
  Vincent Neiger was with Technical University of Denmark, Kgs. Lyngby,
  Denmark, with funding from the People Programme (Marie Curie Actions) of the
  European Union's Seventh Framework Programme (FP7/2007-2013) under REA grant
agreement number 609405 (COFUNDPostdocDTU).}
\affiliation{%
  \institution{{\normalsize Univ. Limoges, CNRS, XLIM, UMR\,7252}}
  \city{F-87000 Limoges} 
  \state{France} 
}
\email{vincent.neiger@unilim.fr}

\author{Johan Rosenkilde}
\affiliation{%
  \institution{Technical University of Denmark}
  \city{Kgs. Lyngby} 
  \state{Denmark} 
}
\email{jsrn@jsrn.dk}

\author{Grigory Solomatov}
\affiliation{%
  \institution{Technical University of Denmark}
  \city{Kgs. Lyngby} 
  \state{Denmark} 
}
\email{grigorys93@gmail.com}

\copyrightyear{2018}
\acmYear{2018}
\setcopyright{acmlicensed}
\acmConference[ISSAC '18]{2018 ACM International Symposium on Symbolic and Algebraic Computation}{July 16--19, 2018}{New York, NY, USA
  \\ \phantom{bla} \\ \phantom{bla} \\ \phantom{bla} }
\acmBooktitle{ISSAC '18: 2018 ACM International Symposium on Symbolic and Algebraic Computation, July 16--19, 2018, New York, NY, USA}
\acmPrice{15.00}
\acmDOI{10.1145/3208976.3208988}
\acmISBN{978-1-4503-5550-6/18/07}

\begin{abstract}
  We consider the computation of two normal forms for matrices over the
  univariate polynomials: the Popov form and the Hermite form. For matrices
  which are square and nonsingular, deterministic algorithms with satisfactory
  cost bounds are known. Here, we present deterministic, fast algorithms for
  rectangular input matrices. The obtained cost bound for the Popov form
  matches the previous best known randomized algorithm, while the cost bound
  for the Hermite form improves on the previous best known ones by a factor
  which is at least the largest dimension of the input matrix.
\end{abstract}


\maketitle

\section{Introduction}
\label{sec:intro}

In this paper we deal with (univariate) polynomial matrices, i.e. matrices in $\polMatSpace[\rdim][\cdim]$ where $\field$ is a field admitting exact computation, typically a finite field.
Given such an input matrix whose row space is the real object of interest, one may ask for a ``better'' basis for the row space, that is, another matrix which has the same row space but also has additional useful properties.
Two important normal forms for such bases are the Popov form \cite{Popov72} and the Hermite form \cite{Hermite1851}, whose definitions are recalled in this paper.
The Popov form has rows which have the minimal possible degrees, while the Hermite form is in echelon form.
A classical generalisation is the \emph{shifted} Popov form of a matrix \cite{BecLab00},
where one incorporates degree weights on the columns: with
zero shift this is the Popov form, while under some extremal shift this becomes
the Hermite form \cite{BeLaVi99}.
We are interested in the efficient computation of these forms, which has
been studied extensively along with the computation of the related but
non-unique reduced forms \cite{Forney75,Kailath80} and weak Popov forms
\cite{MulSto03}.

Hereafter, complexity estimates count basic arithmetic operations in $\field$ on an algebraic RAM, and asymptotic cost bounds omit factors that are logarithmic in the input parameters, denoted by $\softO{\cdot}$.
We let $2 \le \expmatmul \le 3$ be an exponent for matrix multiplication: two matrices in $\matSpace[\rdim]$ can be multiplied in $O(\rdim^\expmatmul)$ operations.
As shown in \cite{CanKal91}, the multiplication of two polynomials in $\polRing$ of degree at most $d$ can be done in $\softO{d}$ operations,
and more generally the multiplication of two polynomial matrices in $\polMatSpace[\rdim][\rdim]$ of degree at most $d$ uses $\softO{\rdim^\expmatmul d}$ operations.

Consider a square, nonsingular $\pmat \in \polMatSpace[\rdim]$ of degree $d$.
For the computation of a reduced form of $\pmat$, the complexity
$\softO{\rdim^\expmatmul d}$ was first achieved by a Las Vegas algorithm of
Giorgi et al.~\cite{GiJeVi03}.  All the subsequent work mentioned in the next
paragraph achieved the same cost bound, which was taken as a target: up to
logarithmic factors, it is the same as the cost for multiplying two matrices
with dimensions and degree similar to those of $\pmat$.

The approach of \cite{GiJeVi03} was de-randomized by Gupta et
al.~\cite{GuSaStVa12}, while Sarkar and Storjohann \cite{SarSto11} showed how
to compute the Popov form from a reduced form; combining these results gives a
deterministic algorithm for the Popov form.  Gupta and Storjohann
\cite{Gupta11,GupSto11} gave a Las Vegas algorithm for the Hermite form; a
Las Vegas method for computing the shifted Popov form for any shift was
described in \cite{Neiger16}.  Then, a deterministic Hermite form algorithm was
given by Labahn et al.~\cite{LaNeZh17}, which was one ingredient in a
deterministic algorithm due to Neiger and Vu \cite{NeiVu17} for the arbitrary
shift case.

The Popov form algorithms usually exploit the fact that, by definition, this
form has degree at most $d=\deg(\pmat)$.  While no similarly strong degree
bound holds for shifted Popov forms in general (including the Hermite form),
these forms still share a remarkable property in the square, nonsingular case:
each entry outside the diagonal has degree less than the entry on the diagonal
in the same column.  These diagonal entries are called \emph{pivots}
\cite{Kailath80}.  Furthermore, their degrees sum to $\deg(\det(\pmat)) \leq
\rdim d$, so that these forms can be represented with $\bigO{\rdim^2 d}$ field
elements, just like $\pmat$.  This is especially helpful in the design of fast
algorithms since this provides ways to control the degrees of the manipulated
matrices.

These degree constraints exist but become weaker in the case of
\emph{rectangular} shifted Popov forms, say $\rdim\times\cdim$ with
$\rdim<\cdim$.  Such a normal form does have $\rdim$ columns containing pivots,
whose average degree is at most the degree $d$ of the input matrix $\pmat$.
Yet it also contains $\cdim-\rdim$ columns without pivots, which may all have
large degree: up to $\Theta(\rdim d)$ in the case of the Hermite form.  As a
result, a dense representation of the latter form may require $\Omega(\rdim^2
(\cdim-\rdim) d)$ field elements, a factor of $\rdim$ larger than for $\pmat$.
Take for example some $\umat \in \polMatSpace$ of degree $d$ which is
unimodular, meaning that $\umat^{-1}$ has entries in $\polRing$.  Then, the
Hermite form of $[\umat \;\; \idMat \;\; \cdots \;\; \idMat]$ is $[\idMat \;\;
\umat^{-1} \;\; \cdots \;\; \umat^{-1}]$, and the entries of $\umat^{-1}$ may
have degree in $\Omega(\rdim d)$.  However, the Popov form, having minimal
degree, has size in $\bigO{\rdim\cdim d}$, just like $\pmat$.  Thus, unlike in
the nonsingular case, one would set different target costs for the
computation of Popov and Hermite forms, such as
$\softO{\rdim^{\expmatmul-1}\cdim d}$ for the former and
$\softO{\rdim^{\expmatmul}\cdim d}$ for the latter (note that the exponent
affects the small dimension).

For a rectangular matrix $\pmat\in\pmatSpace$, Mulders and Storjohann
\cite{MulSto03} gave an iterative Popov form algorithm which costs
$\bigO{\matrk \rdim \cdim d^2}$, where $r$ is the rank of $\pmat$.  Beckermann
et al.~\cite{BeLaVi06} obtain the shifted Popov form for any shift by computing
a basis of the left kernel of $\trsp{[\trsp{\pmat} \;\; \idMat[\cdim]]}$.  This
approach also produces a matrix which transforms $\pmat$ into its normal form
and whose degree can be in $\Omega(\rdim d)$: efficient algorithms usually
avoid computing this transformation.  To find the sought kernel basis, the
fastest known method is to compute a shifted Popov approximant basis of the
$(\rdim+\cdim)\times\cdim$ matrix above, at an order which depends on the
shift.  \cite{BeLaVi06} relies on a fraction-free algorithm for the latter
computation, and hence lends itself well to cases where $\field$ is not finite.
In our context, following this approach with the fastest known approximant
basis algorithm \cite{JeNeScVi16} yields the cost bounds
$\softO{(\rdim+\cdim)^{\expmatmul-1} \cdim \rdim d}$ for the Popov form and
$\softO{(\rdim+\cdim)^{\expmatmul-1} \cdim^2 \rdim d}$ for the Hermite form.
For the latter this is the fastest existing algorithm, to the best of our
knowledge.

For $\pmat$ with full rank and $\rdim\le\cdim$, Sarkar \cite{Sarkar11} showed a
Las Vegas algorithm for the Popov form achieving the cost
$\softO{\rdim^{\expmatmul-1} \cdim d}$.  This uses random column operations to
compress $\pmat$ into an $\rdim\times\rdim$ matrix, which is then transformed
into a reduced form.  Applying the same transformation on $\pmat$ yields a
reduced form of $\pmat$ with high probability, and from there the Popov form can be
obtained.  Lowering this cost further seems difficult, as indicated in the
square case by the reduction from polynomial matrix multiplication to Popov
form computation described in \cite[Thm.\,22]{SarSto11}.

For a matrix $\pmat \in \pmatSpace$ which is rank-deficient or has
$\rdim>\cdim$, the computation of a basis of the row space of $\pmat$ was
handled by Zhou and Labahn \cite{ZhoLab14} with cost
$\softO{\rdim^{\expmatmul-1}(\rdim+\cdim) d}$.  Their algorithm is
deterministic, and the output basis $\bmat\in\polMatSpace[\matrk][\cdim]$ has
degree at most $d$.  This may be used as a preliminary step: the normal form of
$\pmat$ is also that of $\bmat$, and the latter has full rank with
$\matrk\le\cdim$.

We stress that, from a rectangular matrix $\pmat\in\pmatSpace$, it seems
difficult in general to predict which columns of its shifted Popov form will be
pivot-free.  For this reason, there seems to be no obvious deterministic
reduction from the rectangular case to the square case, even when $\cdim$ is
only slightly larger than $\rdim$.  Sarkar's algorithm is a Las Vegas
reduction, \emph{compressing} the matrix to a nonsingular $\rdim\times\rdim$
matrix; another Las Vegas reduction consists in \emph{completing} the matrix to
a nonsingular $\cdim\times\cdim$ matrix (see \cref{sec:fast_lasvegas}).

In the nonsingular case, exploiting information on the pivots has led to
algorithmic improvements for normal form algorithms
\cite{GupSto11,SarSto11,JeNeScVi16,LaNeZh17}.  Following this, we put our
effort into two computational tasks: finding the location of the pivots in the
normal form (the \emph{pivot support}), and using this knowledge to compute
this form.

Our first contribution is to show how to efficiently find the pivot support of
$\pmat$. For this we resort to the so-called saturation of $\pmat$ computed in
a form which reveals the pivot support (\cref{sub:find_pivsupp_viafactor}),
making use of an idea from \cite{ZhoLab13}.  While this is only efficient for
$\cdim\in\bigO{\rdim}$, using this method repeatedly on well-chosen submatrices
of $\pmat$ with about $2\rdim$ columns allows us to find the pivot support
using $\softO{\rdim^{\expmatmul-1} \cdim d}$ operations for any dimensions
$\rdim \le \cdim$ (\cref{sub:wide_matrix}).

In our second main contribution, we consider the shifted Popov form of $\pmat$,
for any shift.  We show that once its pivot support is known, then this form
can be computed efficiently (\cref{sec:known_pivotsupport} and
\cref{prop:algo:knownsupp_popov}).  In particular, combining both contributions
yields a fast and deterministic Popov form algorithm.

\begin{theorem}
  For a matrix $\pmat \in \pmatSpace$ of degree at most $d$ and with
  $\rdim\le\cdim$, there is a deterministic algorithm which computes the Popov
  form of~$\pmat$ using $\softO{\rdim^{\expmatmul-1} \cdim d}$ operations in
  $\field$.
\end{theorem}

The second contribution may of course be useful in situations where the pivot
support is known for some reason.  Yet, there are even general cases where it
can be computed efficiently, namely when the shift has very unbalanced entries.
This is typically the case of the Hermite form, for which the pivot support
coincides with the column rank profile of $\pmat$.  The latter can be
efficiently obtained via an algorithm due to Zhou \cite[Sec.\,11]{Zhou12},
based on the kernel basis algorithm from \cite{ZhLaSt12}.  This leads us to the
next result.

\begin{theorem}
  \label{thm:hermite}
  Let $\pmat \in \pmatSpace$ with full rank and $\rdim<\cdim$.  There is a
  deterministic algorithm which computes the Hermite form of $\pmat$ using
  $\softO{\rdim^{\expmatmul-1} \cdim \degBd}$ operations in $\field$, where
  $\degBd$ is the minimum of the sum of column degrees of $\pmat$ and of the
  sum of row degrees of $\pmat$.
\end{theorem}

Using this quantity $\degBd$ (see \cref{eqn:degBd_hermite} for a more precise
definition), the mentioned cost for the kernel basis approach of
\cite{BeLaVi06} becomes $\softO{(\rdim+\cdim)^{\expmatmul-1} \cdim^2 \degBd}$.
Thus, when $\cdim\in\bigO{\rdim}$ the cost in the above theorem already gains a
factor $\cdim$ compared to this approach; when $\cdim$ is large compared to
$\rdim$, this factor becomes $\cdim (\frac{\cdim}{\rdim})^{\expmatmul-1}$.

\section{Preliminaries}
\label{sec:preliminaries}

\subsection{Basic notation}

If $\pmat$ is an $\rdim\times\cdim$ matrix and $1 \le j \le \cdim$, we denote
by $\matcol{\pmat}{j}$ the $j$th column of $\pmat$.  If $J \subseteq
\{1,\ldots,\cdim\}$ is a set of column indices, $\matcols{\pmat}{J}$ is the
submatrix of $\pmat$ formed by the columns at the indices in $J$.  We use
analogous row-wise notation.  Similarly, for a tuple $\shiftt \in \ZZ^\cdim$,
then $\subTuple{\shiftt}{J}$ is the subtuple of $\shiftt$ formed by the entries
at the indices in $J$.

When adding a constant to an integer tuple, for example $\shiftt + 1$ for some
$\shiftt = (t_1,\ldots,t_m) \in \ZZ^m$, we really mean $(t_1+1,\ldots,t_m+1)$;
when comparing a tuple to a constant, for example $\shiftt \le 1$, we mean
$\max(\shiftt) \le 1$.  Two tuples of the same length will always be compared
entrywise: $\shifts \leq \shiftt$ stands for $s_i \leq t_i$ for all $i$.  We
use the notation $\amp[\shiftt] = \max(\shiftt) - \min(\shiftt)$, and
$\sumVec{\shiftt} = t_1 + \ldots + t_m$ (note that the latter will mostly be
used when $\shiftt$ has nonnegative entries). 

For a given nonnegative integer tuple $\shiftt = (t_1,\ldots,t_m) \in \NN^m$,
we denote by $\shiftMat{\shiftt}$ the diagonal matrix with entries
$\var^{t_1},\ldots,\var^{t_m}$.

\subsection{Row spaces, kernels, and approximants}

For a matrix $\pmat \in \pmatSpace$, its \emph{row space} is the
$\polRing$-module generated by its rows, that is, $\{\rowgrk{\lambda} \pmat ,
\rowgrk{\lambda} \in \polMatSpace[1][\rdim]\}$.  Then, a matrix $\bmat \in
\polMatSpace[\matrk][\cdim]$ is a \emph{row basis} of $\pmat$ if its rows form
a basis of the row space of $\pmat$, in which case $\matrk$ is the rank of
$\pmat$.

The \emph{left kernel} of $\pmat$ is the $\polRing$-module $\{ \row{p} \in
\polMatSpace[1][\rdim] \mid \row{p} \pmat = \matz \}$.  A matrix $\kerbas \in
\polMatSpace[k][\rdim]$ is a \emph{left kernel basis of $\pmat$} if its rows
form a basis of this kernel, in which case $k=\rdim-\matrk$.  Similarly, a
\emph{right} kernel basis of $\pmat$ is a matrix $\kerbas \in
\polMatSpace[\cdim][(\cdim-\matrk)]$ whose \emph{columns} form a basis of the
right kernel of $\pmat$.

Given $\orders=(\order_1,\ldots,\order_\cdim)\in\ZZp^\cdim$, the set of
\emph{approximants for $\pmat$ at order $\orders$} is the $\polRing$-module of
rank $\rdim$ defined as
\[
  \modApp{\orders}{\pmat} =
  \{ \app \in \polMatSpace[1][\rdim] \mid \app \pmat = \matz \bmod \xDiag{\orders} \}.
\]
The identity $\app \pmat = \matz \bmod \xDiag{\orders}$ means that the $j$th
entry of the vector $\app \pmat \in \polMatSpace[1][\cdim]$ is divisible by
$\var^{\order_j}$, for all $j$.

Two $\rdim \times \cdim$ matrices $\pmat_1$, $\pmat_2$ have the same row space
if and only if they are \emph{unimodularly equivalent}, that is, there is a
unimodular matrix $\umat \in \polMatSpace$ such that $\umat \pmat_1 = \pmat_2$.
For $\pmat_3 \in \polMatSpace[\matrk][\cdim]$ with $r \leq \rdim$, $\pmat_1$
and $\pmat_3$ have the same row space exactly when $\pmat_3$ padded with
$\rdim-r$ zero rows is unimodularly equivalent to $\pmat_1$.

\subsection{Row degrees and reduced forms}
\label{ssec:preliminaries:rdeg_reduced}

For a matrix $\pmat \in \pmatSpace$, we denote by $\rdeg{\pmat}$ the tuple of
the degrees of its rows, that is,
$(\deg(\matrow{\pmat}{1}),\ldots,\deg(\matrow{\pmat}{\rdim}))$.

If $\pmat$ has no zero row, the \emph{(row-wise) leading matrix} of $\pmat$,
denoted by $\leadingMat[]{\pmat}$, is the matrix in $\matSpace[\rdim][\cdim]$
whose entry $i,j$ is equal to the coefficient of degree
$\deg(\matrow{\pmat}{i})$ of the entry $i,j$ of $\pmat$.

For a matrix $\reduced \in \pmatSpace$ with no zero row and $\rdim \leq \cdim$,
we say that $\reduced$ is \emph{(row) reduced} if $\leadingMat[]{\reduced}$ has
full rank.  Thus, here a reduced matrix must have full rank (and no zero row),
as in \cite{Forney75}.  For more details about reduced matrices, we refer the
reader to \cite{Wolovich74,Forney75,Kailath80,BeLaVi06}.  In particular, we
have the following characterizing properties:
\begin{itemize}
  \item \emph{Predictable degree property \cite{Forney75}
    \cite[Thm.~6.3-13]{Kailath80}:} we have
    \[
      \deg(\rowgrk{\lambda}\reduced) =
      \max \{ \deg(\lambda_i) + \rdeg{\matrow{\reduced}{i}} , \; 1 \le i \le \rdim \}
    \]
    for any vector $\rowgrk{\lambda} = [\lambda_i]_i \in
    \polMatSpace[1][\rdim]$.
  \item \emph{Minimality of the sum of row degrees \cite{Forney75}:}
    for any nonsingular matrix $\mat{U} \in \polMatSpace[\rdim][\rdim]$, we
    have $\sumVec{\rdeg{\mat{U} \reduced}} \ge \sumVec{\rdeg{\reduced}}$.
  \item \emph{Minimality of the tuple of row degrees \cite[Sec.\,2.7]{Zhou12}:} 
    for any nonsingular matrix $\mat{U} \in \polMatSpace[\rdim][\rdim]$, we
    have $\shifts \le \shifts[t]$ where the tuples $\shifts$ and $\shifts[t]$
    are the row degrees of $\reduced$ and of $\mat{U}\reduced$ sorted in
    nondecreasing order, respectively.
\end{itemize}

From the last item, it follows that two unimodularly equivalent reduced
matrices have the same row degree up to permutation.

For a matrix $\pmat \in \pmatSpace$, we call \emph{reduced form of $\pmat$} any
reduced matrix $\reduced \in \polMatSpace[\matrk][\cdim]$ which is a row basis
of $\pmat$.  The third item above shows that $\deg(\reduced) \le \deg(\pmat)$.

\subsection{Pivots and Popov forms}
\label{ssec:preliminaries:pivots_popov}

For a nonzero vector $\row p = [p_j]_j \in \polMatSpace[1][\rdim]$, the \emph{pivot index} of $\row p$ is the largest index $j$ such that $\deg(p_j) = \deg(\row p)$ \cite[Sec.\,6.7.2]{Kailath80}.
In this case we call $p_j$ the \emph{pivot entry} of $\row p$.
For the zero vector, we define its degree to be $-\infty$ and its pivot index to be $0$.
Further, the \emph{pivot index} of a matrix $\pmat \in \pmatSpace$ is the tuple $(j_1,\ldots,j_m) \in \ZZ_{\geq 0}^m$ such that $j_i$ is the pivot index of $\matrows{\pmat}{i}$.
Note that we will only use the word ``pivot'' in this row-wise sense.

A matrix $\popov \in \pmatSpace$ is in \emph{weak Popov form} if it has no zero row and the entries of the pivot index of $\popov$ are all distinct \cite{MulSto03};
a weak Popov form is further called \emph{ordered} if its pivot index is in (strictly) increasing order.
A weak Popov matrix is also reduced.

The (ordered) weak Popov form is not canonical: a given row space may have many (ordered) weak Popov forms.
The Popov form adds a normalization property, yielding a canonical form;
we use the definition from \cite[Def.\,3.3]{BeLaVi99}:

A matrix $\mat{P} \in \pmatSpace$ is in \emph{Popov form} if it is in ordered weak Popov form, the corresponding pivot entries are monic, and in each column of $\mat{P}$ which contains a (row-wise) pivot the
other entries have degree less than this pivot entry.

For a matrix $\pmat \in \pmatSpace$ of rank $\matrk$, there exists
a unique $\popov \in \polMatSpace[\matrk][\cdim]$ which is in Popov form and has the same row space as $\pmat$ \cite[Thm.\,2.7]{BeLaVi06}.
We call $\popov$ \emph{the Popov form of $\pmat$}.
For a more detailed treatment of Popov forms, see \cite{Kailath80,BeLaVi99,BeLaVi06}.

For example, consider the unimodularly equivalent matrices
\[  
  \begin{bmatrix}
    \var^2 & \var+1 & 2 \\
    2\var+2 & 2\var & 2
  \end{bmatrix}
  \quad\text{and}\quad
  \begin{bmatrix}
    \var^2-\var-1 & 1 & 1 \\
    \var+1 & \var & 1
  \end{bmatrix},
\]
defined over $\mathbb{F}_7[\var]$; the first one is in weak Popov
form and the second one is its Popov form.  Note that any
deterministic rule for ordering the rows would lead to a canonical form; we use
that of \cite{BeLaVi99,BeLaVi06}, while that of \cite{Kailath80,MulSto03} sorts
the rows by degrees and would consider the second matrix not to be normalized.

Going back to the general case, we denote by $\pivSupp[][\pmat]\in
\ZZp^{\matrk}$ the pivot index of the Popov form of $\pmat$, called the
\emph{pivot support} of $\pmat$.  In most cases, $\pivSupp[][\pmat]$ differs
from the pivot index of $\pmat$.  We have the following important properties:
\begin{itemize}
  \item The pivot index of $\pmat$ is equal to the pivot support
    $\pivSupp[][\pmat]$ if and only if $\pmat$ is in ordered weak Popov form.
  \item For any $\rowgrk{\lambda} \in \polMatSpace[1][\rdim]$ such that
    $\rowgrk{\lambda}\pmat \neq \matz$, the pivot index of
    $\rowgrk{\lambda}\pmat$ appears in the pivot support $\pivSupp[][\pmat]$;
    in particular each nonzero entry of the pivot index of $\pmat$ is in
    $\pivSupp[][\pmat]$.
\end{itemize}

For the first item, we refer to \cite[Sec.\,2]{BeLaVi06} (in this reference,
the set formed by the entries of the pivot support is called ``pivot set'' and
ordered weak Popov forms are called quasi-Popov forms).  The second item is a
simple extension of the predictable degree property (see for example
\cite[Lem.\,1.17]{Neiger16b} for a proof).

\subsection{Computational tools}

We will rely on the following result from \cite[Cor.\,4.6 and
Thm.\,3.4]{ZhLaSt12} about the computation of kernel bases in reduced form.
Note that a matrix is \emph{column reduced} if its transpose is reduced.

\begin{theorem}[\cite{ZhLaSt12}]
  \label{thm:kernel_ZhLaSt}
  There is an algorithm \algoname{MinimalKernelBasis} which, given a matrix
  $\pmat \in \polMatSpace[\rdim][\cdim]$ with $\rdim\le\cdim$, returns a right
  kernel basis $\kerbas \in \polMatSpace[\rdim][(\cdim-\matrk)]$ of $\pmat$ in
  column reduced form using
  \[
    \softO{\cdim^\expmatmul \lceil \rdim \deg(\pmat) / \cdim \rceil}
    \subseteq \softO{\cdim^\expmatmul  \deg(\pmat) }
  \]
  operations in $\field$.  Furthermore, $\sumVec{\cdeg{\kerbas}} \le \matrk
  \deg(\pmat)$.
\end{theorem}

For the computation of normal forms of square, nonsingular matrices, we use the
next result ($\shifts$-Popov forms will be introduced in
\cref{sec:preliminaries_shifted}; Popov forms as above correspond to $\shifts =
\unishift$).
\begin{theorem}[\cite{NeiVu17}]
  \label{thm:nonsing_popov_NeiVu}
  There is an algorithm \algoname{NonsingularPopov} which, given a nonsingular
  matrix $\pmat \in \polMatSpace[\rdim]$ and a shift $\shifts \in \ZZ^\rdim$,
  returns the $\shifts$-Popov form of $\pmat$ using
  \[
    \softO{\rdim^{\expmatmul} \lceil\sumVec{\rdeg{\pmat}}/\rdim\rceil }
    \subseteq \softO{\rdim^\expmatmul \deg(\pmat)}
  \]
  operations in $\field$.
\end{theorem}
\noindent This is \cite[Thm.\,1.3]{NeiVu17} with a minor modification: we have
replaced the so-called generic determinant bound by a larger quantity (the sum
of row degrees), since this is sufficient for our needs here.

\section{Popov form via completion into a square and nonsingular matrix}
\label{sec:fast_lasvegas}

We now present a new Las Vegas algorithm for computing the (non-shifted) Popov
form $\popov$ of a rectangular matrix $\pmat \in \pmatSpace$ with full rank and
$\rdim<\cdim$, relying on algorithms for the case of square, nonsingular
matrices.  In the case $\cdim \in O(\rdim)$, this results in a cost bounded by
$\softO{\rdim^\expmatmul \deg(\pmat)}$, which has already been obtained by the
Las Vegas algorithm of Sarkar \cite{Sarkar11}; however, the advantage of our
approach is that it becomes asymptotically faster if the \emph{average} row
degree of $\pmat$ is significantly smaller than $\deg(\pmat)$.

The idea is to find a matrix $\cmat \in \polMatSpace[(\cdim-\rdim)][\cdim]$ such that
the Popov form of $\trsp{[\trsp{\pmat} \;\; \trsp{\cmat}]}$ contains $\popov$ as an identifiable subset of its rows.
We will show that if $\cmat$ is drawn randomly of sufficiently high degree, then this is true with high probability.
\begin{definition}
  \label{dfn:completion}
  Let $\pmat \in \pmatSpace$ have full rank with $\rdim<\cdim$ and let $\popov \in \pmatSpace$ be the Popov form of $\pmat$.
  A \emph{completion of $\pmat$} is any matrix $\cmat \in \polMatSpace[(\cdim-\rdim)][\cdim]$ such that:
  \[
    \min(\rdeg{\cmat}) > \deg(\popov)
    \text{ and }
    \begin{bmatrix} \popov \\ \cmat \end{bmatrix} \text{ is row reduced}.
  \]
\end{definition}

The next lemma shows that: 1) if $\cmat$ is a completion, then $\popov$ will appear as a submatrix of the Popov form of $\trsp{[ \trsp\pmat \;\; \trsp \cmat]}$; and 2) we can easily check from that Popov form whether $\cmat$ is a completion or not.
The latter is essential for a Las Vegas algorithm.

\begin{lemma}
  \label{lem:completion_correct_verifiable}
  Let $\pmat \in \pmatSpace$ have full rank with $\rdim<\cdim$ with Popov form $\popov$, and let
  $\cmat \in \polMatSpace[(\cdim-\rdim)][\cdim]$ be such that $\trsp{[ \trsp\pmat \;\; \trsp \cmat]}$ has full rank and $\min(\rdeg{\cmat}) > \deg(\popov)$.
  Then, $\cmat$ is a completion of $\pmat$ if and only if $\rdeg{\matt P}$ contains a permutation of $\rdeg{\cmat}$, where $\matt{P}$ is the Popov form of $\trsp{[\trsp{\pmat} \;\; \trsp{\cmat}]}$.
  In this case, $\popov$ is the submatrix of $\matt{P}$ formed by its rows of degree less
  than $\min(\rdeg{\cmat})$.
\end{lemma}
\begin{proof}
  First, we assume that $\cmat$ is a completion of $\pmat$.  Then
  $\trsp{[\trsp{\popov} \;\; \trsp{\cmat}]}$ is reduced, and therefore it has
  the same row degree as its Popov form $\matt{P}$ up to permutation.  Hence,
  in particular, $\rdeg{\matt{P}}$ contains a permutation of $\rdeg{\cmat}$. 

  Now, we assume that $\rdeg{\matt{P}}$ contains a permutation of $\rdeg{\cmat}$ and our goal is to
  show that $\trsp{[\trsp{\popov} \;\; \trsp{\cmat}]}$ is reduced and $\matt P$ contains $\popov$ as a submatrix.
  Let $\matt{P}_1$ be the submatrix of $\matt{P}$ of its rows of degree less
  than $\min(\rdeg{\cmat})$; and $\matt{P}_2$ be the submatrix of the remaining rows.
  By assumption, $\matt{P}_2$ has at least $\cdim-\rdim$ rows and $\matt P_1$ has at most $\rdim$ rows.
  Since $\matt{P}$ is also the Popov form of $\trsp{[\trsp{\popov} \;\;
  \trsp{\cmat}]}$, there is a unimodular transformation
  \begin{equation}
    \label{eqn:transfo_completion}
    \left[\begin{array}{cc}
      \umat_{11} & \umat_{12} \\
      \umat_{21} & \umat_{22} \\
    \end{array}\right]
    \left[\begin{array}{c}
      \matt{P}_1 \\
      \matt{P}_2
    \end{array}\right]
    =
    \left[\begin{array}{c}
      \popov \\
      \cmat
    \end{array}\right].
  \end{equation}
  By the predictable degree property we obtain $\umat_{12} = \matz$; thus, since $\popov$ has full rank $\rdim$, then $\matt P_1$ has exactly $\rdim$ rows, and $\umat_{11}$ is unimodular.
  Therefore $\matt P_1 = \popov$ since both matrices are in Popov form.
  As a result, $\rdeg{\matt P}$ is a permutation of $(\rdeg{\popov}, \rdeg{\cmat})$.
\end{proof}

\begin{lemma}
  \label{lem:completion_probability}
  Let $\pmat \in \pmatSpace$ have full rank with $\rdim<\cdim$.
  Let $S \subseteq \field$ be finite of cardinality $q$
  and let $\mat{L} \in \matSpace[(\cdim-\rdim)][\cdim]$ with entries chosen
  independently and uniformly at random from $S$.  Then $\var^{\deg(\pmat)+1}\mat{L}$ is a completion of
  $\pmat$ with probability at least $\prod_{i=1}^{\cdim-\rdim}(1-q^{-i})$ if
  $\field$ is finite and $S=\field$, and at least $1-\frac{\cdim-\rdim}{q}$
  otherwise.
\end{lemma}
\begin{proof}
  Let $d = \deg(\pmat)$. 
  We first note that for $\var^{d+1}\mat{L}$ to be a completion of $\pmat$, it
  is enough that the matrix
  \[
    \leadingMatBig[]{\begin{bmatrix} \popov \\ \cmat \end{bmatrix}}
    = \begin{bmatrix} \leadingMat{\popov} \\ \leadingMat{\cmat} \end{bmatrix}
    = \begin{bmatrix} \leadingMat{\popov} \\ \mat{L} \end{bmatrix}
    \in \matSpace[\cdim]
  \]
  be invertible. Indeed, this implies first that $\trsp{[\trsp{\popov} \;\;
  \trsp{\cmat}]}$ is reduced; and second, that $\cmat$ has no zero row, hence
  $\rdeg{\cmat} = (d+1,\ldots,d+1)$ and $\min(\rdeg{\cmat})=d+1 >
  \deg(\pmat) \ge \deg(\popov)$.

  In the case of a finite field $\field$ with $q$ elements, the probability
  that the above matrix is invertible is $\prod_{i=1}^{\cdim-\rdim}(1-q^{-i})$.
  If $\field$ is infinite or of cardinality $\ge q$, the Schwartz-Zippel lemma
  implies that the probability that the above matrix is singular is at most
  $(\cdim-\rdim) / q$.
\end{proof}

Thus, if $\field$ is infinite, it is sufficient to take $S$ of cardinality at
least $2 (\cdim-\rdim)$ to ensure that $\var^{d+1}\mat{L}$ is a completion with
probability at least $1/2$.  On the other hand, if $\field$ is finite of
cardinality $q$, we have the following bounds on the probability:
\[
  \prod_{i=1}^{\cdim-\rdim}(1-q^{-i})
  >
  \left\{ \begin{array}{ll}
    0.28 & \text{if } q = 2, \\
    0.55 & \text{if } q = 3, \\
    0.75 & \text{if } q > 5.
  \end{array} \right.
\]

In \cref{algo:lasvegas_popov}, we first test the nonsingularity of
$\mat N = \trsp{[\trsp{\pmat} \;\; \trsp{\cmat}]}$ before computing $\matt{P}$, since
the fastest known Popov form algorithms in the square case do not support singular matrices.
Over a field with at least $2n\deg(\mat N) +1$ elements, a simple Monte Carlo test for this is to
evaluate the polynomial matrix at a random $\alpha \in \field$ and testing the resulting scalar
matrix for nonsingularity; this falsely reports singularity only if $\det(\mat N)$ is divisible by $(\var - \alpha)$.
Alternatively, a deterministic check is as follows.  First, apply the partial linearization of
\cite[Sec.\,6]{GuSaStVa12}, yielding a matrix $\expand{\mat{N}} \in \polMatSpace[\expand{\cdim}]$
such that $\expand{\mat{N}}$ is nonsingular if and only if $\mat{N}$ is nonsingular;
$\expand{\cdim} \in \bigO{\cdim}$; and
$\deg(\expand{\mat{N}}) \le \lceil \sumVec{\rdeg{\mat{N}}}/\cdim \rceil$.
This does not involve arithmetic operations.
Since $\expand{\mat{N}}$ is nonsingular if and only if its kernel is trivial, it then remains to compute a kernel basis via the algorithm in \cite{ZhoLab12}, using
$\softO{\cdim^\expmatmul \deg(\expand{\mat{N}})} \subseteq \softO{\cdim^\expmatmul
\lceil\sumVec{\rdeg{\mat{N}}}/\cdim\rceil}$ operations in $\field$.
Instead of considering the kernel, one could also test the nonsingularity of
$\expand{\mat{N}}$ using algorithms from \cite{GuSaStVa12}, as explained in
\cite[p.\,24]{Sarkar11}.

\begin{algobox}
  \algoInfo{RandomCompletionPopov}
  \label{algo:lasvegas_popov}

  \dataInfo{Input}{
    matrix $\pmat \in \pmatSpace$ with full rank and $\rdim<\cdim$;
    subset $S \subseteq \field$ of cardinality $q$.
  }

  \dataInfo{Output}{
    the Popov form of $\pmat$, or \algoword{failure}.
  }

  \algoSteps{
    \item $\mat{L} \assign$ matrix in $\matSpace[(\cdim-\rdim)][\cdim]$
      with entries chosen uniformly and independently at random from $S$.
    \item $\cmat \assign \var^{\deg(\pmat)+1} \mat{L}$
    \item \algoword{If} $\trsp{[\trsp{\pmat} \;\; \trsp{\cmat}]}$ is singular \algoword{then} \algoword{return failure}
      \label{step:RCP:singular}
    \item $\matt{\popov} \assign \algoname{NonsingularPopov}(\trsp{[\trsp{\pmat} \;\; \trsp{\cmat}]})$
      \label{step:RCP:nonsingularpopov}
    \item \algoword{If} $\rdeg{\matt{\popov}}$ does not contain a
      permutation of $\rdeg{\cmat}$ \algoword{then return failure}
      \label{step:RCP:baddegs}
    \item \algoword{Return} the submatrix of $\matt{\popov}$ formed by its rows
      of degree less than $\min(\rdeg{\cmat})$
  }
\end{algobox}

\begin{proposition}
  \label{prop:algo:lasvegas_popov}
  \cref{algo:lasvegas_popov} is correct and the probability that a failure is
  reported at \cref{step:RCP:singular} or \cref{step:RCP:baddegs} is as indicated in
  \cref{lem:completion_probability}.  If \algoname{NonsingularPopov} is the
  algorithm of \cite{NeiVu17}, \cref{algo:lasvegas_popov} uses
  \[
    \softO{\cdim^\expmatmul \left\lceil \frac{\sumVec{\rdeg{\pmat}} + (\cdim-\rdim) \deg(\pmat)}{\cdim} \right\rceil}
    \;\subseteq\;
    \softO{\cdim^\expmatmul \deg(\pmat)}
  \]
  operations in $\field$.
\end{proposition}

Indeed, from \cref{thm:nonsing_popov_NeiVu}, \cref{step:RCP:nonsingularpopov}
uses $\softO{\cdim^\expmatmul \lceil \Delta/\cdim \rceil}$ operations where
$\Delta = \sumVec{\rdeg{\trsp{[\trsp{\pmat} \;\; \trsp{\cmat}]}}} =
\sumVec{\rdeg{\pmat}} + (\cdim-\rdim) (\deg(\pmat)+1)$.

While other Popov form algorithms could be used, that of \cite{NeiVu17} allows
us to take into account the average row degree of $\pmat$.  Indeed, if
$\sumVec{\rdeg{\pmat}} \ll \rdim\deg(\pmat)$ and $\cdim-\rdim \ll \cdim$, the
cost bound above is asymptotically better than $\softO{\cdim^\expmatmul
\deg(\pmat)}$.

\begin{remark}
  \label{rmk:completion_knownpiv}
  As we mentioned in \cref{ssec:preliminaries:pivots_popov}, the pivot index of $\pmat$ is a subset of $\pivSupp[]$.
  Therefore, one can let $\mat L$ be zero at all columns where $\pmat$ has a pivot, or indices one otherwise knows appear in $\pivSupp[]$.
  If $\pmat$ has uneven degrees (e.g. it has the form $\pmatt \shiftMat{\shifts}$ for some shift $\shifts$, see \cref{ssec:preliminaries:shifted_forms}), then this can be particularly worthwhile.
  In the case where for some reason we know $\pivSupp[]$, then $\mat L$ can simply be taken such that $\matcols{\mat L}{\{1,\ldots,\cdim\} \setminus \pivSupp[]}$ is the identity matrix.
  In that case, \cref{algo:lasvegas_popov} becomes deterministic.
\end{remark}

\section{Computing the pivot support}
\label{sec:find_pivsupp}

We now consider a matrix $\pmat \in \pmatSpace$ with $\rdim<\cdim$, possibly rank-deficient, and we
focus on the computation of its pivot support $\pivSupp[]$.
In \cref{sub:find_pivsupp_viafactor}, we give a deterministic algorithm which is efficient when
$\cdim \in \bigO{\rdim}$.
In \cref{sub:wide_matrix} we explain how this can be used iteratively to efficiently find the pivot
support when $\rdim \ll \cdim$.

\subsection{Deterministic pivot support computation via column basis factorization}
\label{sub:find_pivsupp_viafactor}

Our approach stems from the fact (see \cref{lem:pivsupp_via_factorization})
that $\pivSupp[]$ is also the pivot support of any basis of the
\emph{saturation} of the row space of $\pmat$
\cite[Sec.\,II.\S2.4]{BourbakiCommAlg2}, defined as
\[
  \{ \rowgrk{\lambda} \pmat , \rowgrk{\lambda} \in \field(\var)^{1\times\rdim} \} \cap \polMatSpace[1][\rdim].
\]
This notion of saturation was already used in \cite{ZhoLab13} in order to
compute column bases of $\pmat$ by relying on the following factorization:

\begin{lemma}[{\cite[Sec.\,3]{ZhoLab13}}]
  \label{lem:colbas_factorization}
  Let $\pmat \in \pmatSpace$ have rank $\matrk\in\ZZp$, let $\kerbas \in
  \polMatSpace[\cdim][(\cdim-\matrk)]$ be a right kernel basis of $\pmat$, and
  let $\rfac \in \polMatSpace[\matrk][\cdim]$ be a left kernel basis of
  $\kerbas$.  Then, we have $\pmat = \colbas \rfac$ for some column basis
  $\colbas \in \polMatSpace[\rdim][\matrk]$ of $\pmat$.
\end{lemma}

One can easily verify that the left kernel of $\kerbas$ is precisely the
saturation of $\pmat$, and therefore the matrix $\rfac$ is a (row) basis of this
saturation.  Here, we are particularly interested in the following consequence
of this result:

\begin{lemma}
  \label{lem:pivsupp_via_factorization}
  The matrices $\pmat$ and $\rfac$ in \cref{lem:colbas_factorization} have the
  same pivot support, that is, $\pivSupp[] = \pivSupp[][\rfac]$.
\end{lemma}
\begin{proof}
  Since $\pmat = \colbas \rfac$, the row space of $\pmat$ is contained in that of $\rfac$.
  Hence, by the properties at the end of \cref{ssec:preliminaries:pivots_popov}, $\pivSupp[] \subseteq \pivSupp[][\rfac]$ as sets.
  But since $\pmat$ and $\rfac$ both have rank $\matrk$, both pivot supports have exactly $\matrk$ different elements, and must be equal.
\end{proof}

We will read off $\pivSupp[][\rfac]$ from $\rfac$ by ensuring that this matrix is in ordered weak Popov form.
First, we obtain a column reduced right kernel basis
$\kerbas$ of $\pmat$
using \algoname{MinimalKernelBasis} (see \cref{thm:kernel_ZhLaSt}).
However, the degree profile of $\kerbas$ prevents us from using the same algorithm to compute a
left kernel basis $\rfac$ efficiently, since the average row degree of $\kerbas$ could be as large
as $\matrk \deg(\pmat)$.
To circumvent this issue, we combine the observations that $\deg(\rfac)$ is bounded and that $\kerbas$ has small average \emph{column} degree to conclude that $\rfac$ can be efficiently obtained via an approximant basis (see \cref{sec:preliminaries}).

\begin{lemma}
  \label{lem:rfac_via_appbas}
  Let $\pmat \in \pmatSpace$ have rank $\matrk\in\ZZp$ and let $\kerbas \in
  \polMatSpace[\cdim][(\cdim-\matrk)]$ be a right kernel basis of $\pmat$.
  Then, any left kernel basis of $\kerbas$ which is in reduced form must have
  degree at most $d=\deg(\pmat)$.  As a consequence, if $\appbas \in
  \polMatSpace[\cdim]$ is a reduced basis of $\modApp{\orders}{\kerbas}$, where
  $\orders = \cdeg{\kerbas}+d+1 \in \ZZ^{n-r}$, then the submatrix $\popov$ of $\appbas$
  formed by its rows of degree at most $d$ is a reduced left kernel basis of
  $\kerbas$.
\end{lemma}
\begin{proof}
  Let $\rfac \in \polMatSpace[\matrk][\cdim]$ be a left kernel basis of
  $\kerbas$ in reduced form.  By \cref{lem:colbas_factorization}, $\pmat =
  \colbas \rfac$ for some matrix $\colbas \in \polMatSpace[\rdim][\matrk]$.
  Then, the predictable degree property implies that $\deg(\rfac) \le
  \deg(\colbas\rfac) = d$.

  For the second claim (which is a particular case of
  \cite[Lem.\,4.2]{ZhoLab13}), note that $\popov$ is reduced as a subset of the
  rows of a reduced matrix.  Besides, $\cdeg{\popov \kerbas} < \orders$ by
  construction, hence $\popov\kerbas = \matz \bmod \shiftMat{\orders}$ implies
  $\popov\kerbas= \matz$.  It remains to show that $\popov$ generates the left
  kernel of $\kerbas$. Indeed, there exists a basis of this kernel which has
  degree at most $d$, and on the other hand any vector of degree at most $d$ in
  this kernel is in particular in $\modApp{\orders}{\kerbas}$ and therefore is
  a combination of the rows of $\appbas$; using the predictable degree
  property, we obtain that this combination only involves rows from the
  submatrix $\popov$.
\end{proof}

If we compute $\appbas$ in ordered weak Popov form, then the submatrix $\popov$ is in ordered weak Popov
form as well, and therefore $\pivSupp[]$ can be directly read off from it.
The computation of an approximant basis in ordered weak Popov form can be done via the algorithm of
\cite{JeNeScVi16}, which returns one in Popov form.

\begin{algobox}
  \algoInfo{PivotSupportViaFactor}
  \label{algo:pivsupp_viafactor}

  \dataInfo{Input}{
    matrix $\pmat \in \pmatSpace$ with $\rdim \le \cdim$.
  }

  \dataInfo{Output}{
    the pivot support $\pivSupp[]$ of $\pmat$.
  }

  \algoSteps{
    \item \algoword{If} $\pmat=\matz$ \algoword{then return} the empty tuple $() \in \ZZp^0$
    \item $\kerbas \in \polMatSpace[\cdim][(\cdim-\matrk)] \assign \algoname{MinimalKernelBasis}(\pmat)$
      \label{step:PSVF:kernel}
    \item $\appbas \in \polMatSpace[\cdim] \assign$ ordered weak Popov basis of
      $\modApp{\orders}{\kerbas}$, with $\orders =
      \cdeg{\kerbas}+(\deg(\pmat)+1) \in \ZZ^{\cdim-\matrk}$
      \label{step:PSVF:wPf}
    \item $\rfac \in \polMatSpace[\matrk][\cdim] \assign$ the rows of $\appbas$ of degree at most $\deg(\pmat)$
    \item \algoword{Return} the pivot index of $\rfac$
      \label{step:PSVF:return}
  }
\end{algobox}

\begin{proposition}
  \label{prop:algo:pivsupp_viafactor}
  \cref{algo:pivsupp_viafactor} is correct and uses $\softO{\cdim^\expmatmul
  \deg(\pmat)}$ operations in $\field$.
\end{proposition}
\begin{proof}
  Note that we compute the rank of $\pmat$ as $\matrk$ by the indirect assignment at \cref{step:PSVF:kernel}.
  Besides, $\rfac$ is in ordered weak Popov form since it is a submatrix formed
  by rows of $\appbas$ itself in ordered weak Popov form.  This implies that
  \cref{step:PSVF:return} indeed returns the pivot support of $\rfac$.
  Then, the correctness
  directly follows from \cref{lem:pivsupp_via_factorization,lem:rfac_via_appbas}.

  By \cref{thm:kernel_ZhLaSt}, \cref{step:PSVF:kernel} costs
  $\softO{\cdim^\expmatmul d}$, where $d = \deg(\pmat)$, and
  $\sumVec{\cdeg{\kerbas}} \le \matrk d$.  Thus, the sum of the approximation
  order defined at \cref{step:PSVF:wPf} is $\sumVec{\orders} =
  \sumVec{\cdeg{\kerbas}} + (\cdim-\matrk) (d+1) < \cdim (d+1)$. Then,
  this step uses $\softO{\cdim^{\expmatmul-1} \sumVec{\orders}} \subseteq
  \softO{\cdim^\expmatmul d}$ operations \cite[Thm.\,1.4]{JeNeScVi16}.
\end{proof}

Note that in this algorithm we do not require that $\pmat$ has full rank.
The only reason why we assume $\rdim \le \cdim$ is because the cost bound for the computation of a
kernel basis at \cref{step:PSVF:kernel} is not clear to us in the case $\rdim>\cdim$ (the same
assumption is made in \cite{ZhLaSt12}).

Here, it seems more difficult to take average degrees into account than in
\cref{algo:lasvegas_popov}.  While the average degree of the $\rdim$ columns of
$\pmat$ with largest degree could be taken into account by the kernel basis
algorithm of \cite{ZhLaSt12}, it seems that the computation of $\rfac$ via an
approximant basis remains in $\softO{\cdim^\expmatmul d}$ nevertheless.

\subsection{The case of wide matrices}
\label{sub:wide_matrix}

In this section we will deal with pivots of submatrices $\matcols{\pmat}{J}$, where $J = \{ j_1 < \ldots < j_k \} \subseteq \{ 1,\ldots,n \}$.
To use column indices of $\matcols{\pmat}{J}$ in $\pmat$, we introduce for any such $J$ the operator $\phi_J: \{1,\ldots,k\} \rightarrow \{1,\ldots,n\}$ satisfying $\phi_J(i) = j_i$.
We abuse notation by applying $\phi_J$ element-wise to tuples, such as in $\phi_J(\pivSupp[][\matcols{\pmat}{J}])$.

The following simple lemma is the crux of the algorithm:
\begin{lemma}
  \label{lem:submatrix_popov}
  Let $\pmat \in \pmatSpace$, and consider any set of indices $J \subseteq \{1,\ldots,\cdim\}$.
  Then $(\pivSupp[] \cap J) \subseteq \phi_J(\pivSupp[][\matcols{\pmat}{J}])$ with equality whenever $\pivSupp[] \subseteq J$.
\end{lemma}
\begin{proof}
  If a vector $\row v \in \polMatSpace[1][\cdim]$ in the row space of $\pmat$ is such that $\pivSupp[][\row v] \in J$, then $\pivSupp[][\row v] = \phi_J(\pivSupp[][\matcols{\row v}{J}])$.
  This implies $(\pivSupp[] \cap J) \subseteq \phi_J(\pivSupp[][\matcols{\pmat}{J}])$ since the pivot index of any vector in the row space of $\pmat$ (resp.~$\matcols{\pmat}{J}$) appears in $\pivSupp[]$ (resp.~$\pivSupp[][\matcols{\pmat}{J}]$), see \cref{ssec:preliminaries:pivots_popov}.
  It also immediately implies the equality whenever $\pivSupp[][\pmat] \subseteq J$.
\end{proof}

These properties lead to a fast method for computing the pivot support when
$\cdim \gg \rdim$, relying on a black box \algoname{PivotSupport} which
efficiently finds the pivot support when $\cdim \in \bigO{\rdim}$:
one first considers the $2\rdim$ left columns
$\matcols{\pmat}{\{1,\ldots,2\rdim\}}$ and uses \algoname{PivotSupport} to
compute their pivot support $\tuple{\piv}_1$.  Then, \cref{lem:submatrix_popov}
suggests to discard all columns of $\pmat$ in $\{1,\ldots,2\rdim\}\setminus\tuple{\piv}_1$, thus
obtaining a matrix $\pmat_1$. Then, we repeat the same process to obtain
$\pmat_2,\pmat_3,$ etc.

\begin{algobox}
  \algoInfo{WideMatrixPivotSupport}
  \label{algo:widemat_findsupp}

  \dataInfo{Input}{
    matrix $\pmat \in \pmatSpace$.
  }

  \dataInfo{Output}{
    the pivot support $\pivSupp[]$ of $\pmat$.
  }

  \dataInfo{Assumption}{the algorithm \algoname{PivotSupport} takes as input $\pmat$ and returns
    $\pivSupp[]$.
  }

  \algoSteps{
    \item \algoword{If} $\cdim \le 2\rdim$ \algoword{then return} $\algoname{PivotSupport}(\pmat)$
    \item $\tuple\piv_0 \assign \algoname{PivotSupport}(\matcols{\pmat}{\{1,\ldots,2\rdim\}})$
    \item $\pmatt \assign [\matcols{\pmat}{\tuple\piv_0} \;\;\; \matcols{\pmat}{\{2\rdim+1,\ldots,\cdim\}}]$
    \item $[\tuple\piv_1 \;\; \tuple\piv_2 ] \assign \algoname{WideMatrixPivotSupport}(\pmatt)$, \\
      such that $\max(\tuple\piv_1) \leq \#\tuple\piv_0$ and $\min(\tuple\piv_2) > \#\tuple\piv_0$.
    \item \algoword{Return} $\left[ \phi_{\tuple\piv_0}(\tuple\piv_1) \;\;\; \phi_{\{2m+1,\ldots,n\}}(\tuple\piv_2) \right]$
  }
\end{algobox}

\begin{proposition}
  \label{prop:algo:widemat_findsupp}
  \cref{algo:widemat_findsupp} is correct. It uses at most $\lceil \cdim /
  \rdim \rceil$ calls to \algoname{PivotSupport}, each with a $\rdim \times k$
  submatrix of $\pmat$ as input, where $k \le 2 \rdim$.  If $\rdim\le\cdim$ and
  \algoname{PivotSupport} is \cref{algo:pivsupp_viafactor}, then
  \cref{algo:widemat_findsupp} uses $\softO{\rdim^{\expmatmul-1} \cdim
  \deg(\pmat)}$ operations in $\field$.
\end{proposition}
\begin{proof}
  The correctness follows from \cref{lem:submatrix_popov}, and the operation
  count is obvious.  If using \cref{algo:pivsupp_viafactor} for
  \algoname{PivotSupport}, the correctness and cost bound follow from
  \cref{prop:algo:pivsupp_viafactor}.
\end{proof}

\section{Preliminaries on shifted forms}
\label{sec:preliminaries_shifted}

\subsection{Shifted forms}
\label{ssec:preliminaries:shifted_forms}

The notions of reduced and Popov forms presented in
\cref{ssec:preliminaries:rdeg_reduced,ssec:preliminaries:pivots_popov} can be extended by
introducing additive integer weights in the degree measure for vectors, following
\cite[Sec.\,3]{BarBul92}: a \emph{shift} is a tuple
$\shifts = (\shift{1},\ldots,\shift{\cdim}) \in \ZZ^\cdim$, and the \emph{shifted degree} of a row
vector $\row{p} = [p_1 \; \cdots \; p_\cdim] \in \polMatSpace[1][\cdim]$ is
\[
  \rdeg[\shifts]{\row{p}} = \max(\deg(p_1)+\shift{1},
  \ldots,\deg(p_\cdim)+\shift{\cdim})
  = \rdeg{\row{p}\shiftMat{\shifts}},
\]
where $\shiftMat{\shifts} =
\diag{\var^{\shift{1}},\ldots,\var^{\shift{\cdim}}}$.  Note that here
$\row{p}\shiftMat{\shifts}$ may be over the ring of Laurent polynomials
if $\min(\shifts) < 0$; below, actual computations will always remain over
$\polRing$.
Note that with $\shifts=\tuple 0$ we recover the notion of degree used in the previous sections.

This leads to shifted reduced forms for cases where one is interested in matrices whose rows minimize the $\shifts$-degree, instead of the usual $\unishift$-degree.
The generalized definitions from \cref{sec:preliminaries} can be
concisely described as follows.  For a matrix $\pmat \in
\polMatSpace[\rdim][\cdim]$, its $\shifts$-row degree is $\rdeg[\shifts]{\pmat}
= \rdeg{\pmat\shiftMat{\shifts}}$.  If $\pmat$ has no zero row, its
$\shifts$-leading matrix is $\leadingMat[\shifts]{\pmat} =
\leadingMat[]{\pmat\shiftMat{\shifts}}$, and the $\shifts$-pivot index and
entries of $\pmat$ are the pivot index and entries of
$\pmat\shiftMat{\shifts}$.  The \emph{$\shifts$-pivot degree} of $\pmat$ is the
tuple of the degrees of its $\shifts$-pivot entries; this is equal to
$\rdeg[\shifts]{\pmat} - \subTuple{\shifts}{J}$, where $J$ is the
$\shifts$-pivot index of $\pmat$ and $\subTuple{\shifts}{J}$ the corresponding
subshift.

If $\pmat$ has no zero row and $\rdim\le\cdim$, then $\pmat$ is in
$\shifts$-reduced, $\shifts$-(ordered) weak Popov or $\shifts$-Popov form if
$\pmat\shiftMat{\shifts}$ has the respective non-shifted form, whenever
$\min(\shifts) \geq 0$.  Since adding a constant to all the entries of
$\shifts$ simply shifts the $\shifts$-degree of vectors by this constant, this
does not change the $\shifts$-leading matrix or the $\shifts$-pivots, and thus
does not affect the shifted forms. Therefore we can extend the definitions of
these to also cover $\shifts$ with negative entries; one may alternatively
assume $\min(\shifts)=0$ without loss of generality.

The $\shifts$-Popov form $\popov$ of a matrix $\pmat\in\pmatSpace$ is the
unique row basis of $\pmat$ which is in $\shifts$-Popov form.  The
$\shifts$-pivot support of $\pmat$ is the $\shifts$-pivot index of $\popov$ and
is denoted by $\pivSupp[\shifts][\pmat] \in \ZZp^{\matrk}$, where $\matrk$ is
the rank of $\pmat$.  For more details on shifted forms, we refer to
\cite{BeLaVi06}.

Computationally, it is folklore that finding the shifted Popov form easily
reduces to the non-shifted case: given a matrix $\pmat \in
\polMatSpace[\rdim][\cdim]$ and a nonnegative shift $\shifts \in \ZZ^\cdim$,
the non-shifted Popov form $\matt{P}$ of $\pmat \shiftMat{\shifts}$ has the
form $\matt P = \popov \shiftMat{\shifts}$, with $\popov$ the $\shifts$-Popov
form of $\pmat$. If $\rdim < \cdim$ and the computation of $\matt{P}$ can be
carried out in $\softO{\rdim^{\expmatmul-1} \cdim \deg(\pmat)}$ operations,
this approach yields $\popov$ in $\softO{\rdim^{\expmatmul-1} \cdim
(\deg(\pmat) + \amp)}$.  While this cost is satisfactory whenever $\amp \in
\bigO{\deg(\pmat)}$, one may hope for improvements especially when $\amp >
\rdim \deg(\pmat)$. Indeed, \cref{eqn:bound_deg_popov_global} in
\cref{lem:popov_bounds} shows $\deg(\popov) \leq \rdim \deg(\pmat)$, suggesting
the target cost $\softO{\rdim^{\expmatmul} \cdim \deg(\pmat)}$ for the
computation of $\popov$.

\subsection{Hermite form}
\label{ssec:hermite_def}

A matrix $\hermite = [h_{i,j}] \in \polMatSpace[\matrk][\cdim]$ with
$\matrk\le\cdim$ is in \emph{Hermite form}
\cite{Hermite1851,MacDuffee33,Newman72} if there are indices $1 \le j_1 <
\cdots < j_\matrk \le \cdim$ such that:
\begin{itemize}
  \item $h_{i,j} = 0$ for $1 \le j < j_i$ and $1\le i\le\matrk$,
  \item $h_{i,j_i}$ is monic (therefore nonzero) for $1\le i\le\matrk$,
  \item $\deg(h_{i',j_i})<\deg(h_{i,j_i})$ for $1\le i' < i\le\matrk$.
\end{itemize}
We call $(j_1,\ldots,j_\matrk)$ the \emph{Hermite pivot index} of $\hermite$;
note that it is precisely the column rank profile of $\hermite$.

For a matrix $\pmat \in \pmatSpace$, its Hermite form $\hermite \in
\polMatSpace[\matrk][\cdim]$ is the unique row basis of $\pmat$ which is in
Hermite form.  We call \emph{Hermite pivot support} of $\pmat$ the Hermite
pivot index of $\hermite$.  Note that this is also the column rank profile of
$\pmat$, since $\pmat$ is unimodularly equivalent to $\hermite$ (up to padding
$\hermite$ with zero rows).

For a given $\pmat$, the Hermite form can be seen as a specific shifted Popov
form: defining the shift $\shifts[h] = (\cdim t,\ldots,2t,t)$ for any
$t>\deg(\hermite)$, the $\shifts[h]$-Popov form of $\pmat$ coincides with its
Hermite form \cite[Lem.\,2.6]{BeLaVi06}.  Besides, the $\shifts[h]$-pivot index
of $\hermite$ is $(j_1,\ldots,j_\matrk)$; in other words, the Hermite pivot
support $\pivSupp[\shifts[h]]$ is the column rank profile of $\pmat$.  

\subsection{Degree bounds for shifted Popov forms}
\label{ssec:degree_bounds}

The next result states that the unimodular transformation $\umat$ between
$\pmat$ and its $\shifts$-Popov form $\popov$ only depends on the submatrices
of $\pmat$ and $\popov$ formed by the columns in the $\shifts$-pivot support.
It also gives useful degree bounds for the matrices $\umat$ and $\popov$; for a
more general study of such bounds, we refer to \cite[Sec.\,5]{BeLaVi06}.

\begin{lemma}
  \label{lem:popov_bounds}
  Let $\pmat \in \pmatSpace$ have full rank with $\rdim \le \cdim$, let
  $\shifts\in\shiftSpace$, let $\popov\in\pmatSpace$ be the $\shifts$-Popov
  form of $\pmat$, and let $\tuple\piv = \pivSupp[\shifts][\pmat]$ be the
  $\shifts$-pivot index of $\popov$.  Then $\matcols{\pmat}{\tuple\piv} \in
  \polMatSpace[\rdim]$ is nonsingular, $\matcols{\popov}{\tuple\piv}$ is its
  $\subTuple{\shifts}{\tuple\piv}$-Popov form, and $\umat =
  \matcols{\popov}{\tuple\piv} \matcols{\pmat}{\tuple\piv}^{-1} \in
  \polMatSpace[\rdim]$ is the unique unimodular matrix such that $\umat \pmat =
  \popov$.

  Furthermore, we have the following degree bounds:
  \begin{align}
      \deg(\popov)
      & \le \deg(\pmat) + \amp,
        \label{eqn:bound_deg_popov_shift} \\[0.05cm]
      \cdeg{\matcol{\umat}{i}}
      & \le \sumVec{\rdeg{\pmat}} - \rdeg{\matrow{\pmat}{i}}
         \; \text{ for } 1\le i\le\rdim,
         \label{eqn:bound_cdeg_umat} \\[0.05cm]
      \deg{\umat}
      & \le \sumVec{\cdeg{\matcols{\pmat}{\tuple\piv}}},
        \label{eqn:bound_deg_umat} \\[0.05cm]
      \deg(\popov)
      & \le \min(\sumVec{\rdeg{\pmat}},\sumVec{\cdeg{\pmat'}}) \le \rdim \deg(\pmat) \nonumber \\[-0.1cm]
      & \text{where $\pmat'$ is $\pmat$ with its zero columns removed}.
        \label{eqn:bound_deg_popov_global}
  \end{align}
\end{lemma}

\begin{proof}
  Let $\matt{P} = \matcols{\pmat}{\tuple\piv}$, $\pmatt =
  \matcols{\pmat}{\tuple\piv}$, and $\shiftss=\subTuple{\shifts}{\tuple\piv}$.
  Note first that $\matt{P}$ is nonsingular and in $\shiftss$-Popov form. Let
  $\mat{V} \in \polMatSpace[\rdim]$ be any unimodular matrix such that $\mat{V}
  \pmat = \popov$. Then in particular $\mat{V} \pmatt = \matt{P}$, hence
  $\pmatt$ is nonsingular and unimodularly equivalent to $\matt{P}$, which is
  therefore the $\shiftss$-Popov form of $\pmatt$. Besides, we have $\mat{V} =
  \matt{P} \pmatt^{-1} = \umat$.

  It remains to prove the degree bounds.  The first one comes from the
  minimality of $\popov$.  Indeed, since $\popov$ is an $\shifts$-reduced form
  of $\pmat$ we have $\max(\rdeg[\shifts]{\popov}) \le
  \max(\rdeg[\shifts]{\pmat})$; the left-hand side
  of this inequality is at least $\deg(\popov)+\min(\shifts)$ while its
  right-hand side is at most $\deg(\pmat)+\max(\shifts)$.

  Let $\minDegs \in \NN^\rdim$ be the $\shifts$-pivot degree of $\popov$.
  Then, $\matt{P}$ is in $(-\minDegs)$-Popov form with
  $\rdeg[-\minDegs]{\matt{P}} = \unishift$ and $\cdeg{\matt{P}} = \minDegs$
  \cite[Lem.\,4.1]{JeNeScVi16}.  Besides, $\matt{P}$ is column reduced and thus
  $\sumVec{\cdeg{\matt{P}}} = \deg(\det(\matt{P}))$
  \cite[Sec.\,6.3.2]{Kailath80}, hence $\sumVec{\minDegs} =
  \deg(\det(\pmatt))$.

  Let $\shiftt = (t_1,\ldots,t_m) = \rdeg{\umat^{-1}}$.  We obtain
  $\rdeg[-\minDegs]{\pmatt} = \rdeg[-\minDegs]{\umat^{-1} \matt{P}} =
  \rdeg[\unishift]{\umat^{-1}} = \shiftt$ by the predictable degree property
  (with shifts, see e.g. \cite[Lem.\,2.17]{Zhou12}).  Now, $\umat$ being the
  transpose of the
  matrix of cofactors of $\umat^{-1}$ divided by the constant $\det(\umat^{-1})
  \in \field\setminus\{0\}$, we obtain $\cdeg{\matcol{\umat}{i}} \le
  \sumVec{\shiftt} - t_i$ for $1\le i\le \rdim$. Since $-\minDegs \le
  \unishift$ we have $\shiftt = \rdeg[-\minDegs]{\pmatt} \le \rdeg{\pmat}$,
  hence $\sumVec{\shiftt} - t_i \le \sumVec{\rdeg{\pmat}} -
  \rdeg{\matrow{\pmat}{i}}$. This proves \eqref{eqn:bound_cdeg_umat}.

  Every entry of the adjugate of $\pmatt$ has degree at most
  $\sumVec{\cdeg{\pmatt}}$.  Then, $\umat = \matt{P} \pmatt^{-1}$ gives
  $\deg(\umat) \le \deg(\matt{P}) - \deg(\det(\pmatt)) +
  \sumVec{\cdeg{\pmatt}}$.  This yields \eqref{eqn:bound_deg_umat} since
  $\deg(\matt{P}) = \max(\minDegs) \le \sumVec{\minDegs} = \deg(\det(\pmatt))$.

  The second inequality in \eqref{eqn:bound_deg_popov_global} is implied by
  $\sumVec{\rdeg{\pmat}} \le \rdim \deg(\pmat)$. Besides, from $\popov=
  \umat\pmat=\sum_{i=1}^m \matcol{\umat}{i} \matrow{\pmat}{i}$ we see that
  \eqref{eqn:bound_cdeg_umat} implies $\deg(\popov) \le \sumVec{\rdeg{\pmat}}$.
  For $j\in\tuple\piv$ we have $\cdeg{\matcol{\popov}{j}} \le
  \sumVec{\cdeg{\matt{P}}} = \deg(\det(\pmatt)) \le \sumVec{\cdeg{\pmat'}}$.
  Now, let $j \in \{1,\ldots,\cdim\}\setminus\tuple\piv$: if $\matcol{\pmat}{j}
  = \matz$ then $\matcol{\popov}{j} = \matz$, and otherwise it follows from
  \eqref{eqn:bound_deg_umat} that $\cdeg{\matcol{\popov}{j}} =
  \deg(\umat\matcol{\pmat}{j}) \le \sumVec{\cdeg{\pmatt}} +
  \cdeg{\matcol{\pmat}{j}} \le \sumVec{\cdeg{\pmat'}}$.
\end{proof}

\section{Shifted Popov form when the pivot support is known}
\label{sec:known_pivotsupport}

Now, we focus on computing the $\shifts$-Popov form $\popov$ of $\pmat$ when
the $\shifts$-pivot support $\pivSupp$ is known; here, $\pmat$ has full rank
with $\rdim<\cdim$.

To exploit the knowledge of $\tuple\piv = \pivSupp$, a first approach follows
from \cref{rmk:completion_knownpiv}: use \cref{algo:lasvegas_popov} with
$\mat{L}$ such that $\matcols{\mat L}{\{1,\ldots,\cdim\} \setminus \tuple\piv}$
is the identity matrix and its other columns are zero. Then, it is easily
checked that $\cmat = \mat L \shiftMat{\max(\rdeg[\shifts]\pmat) - \shifts}$ is
a completion of $\pmatt = \pmat \shiftMat{\shifts}$; hence
\cref{algo:lasvegas_popov} returns the Popov form $\matt{P} = \popov
\shiftMat{\shifts}$ of $\pmatt$.  This yields~$\popov$ deterministically in
$\softO{\cdim^\expmatmul(\deg(\pmat) + \amp)}$ operations.

Both factors in this cost bound are unsatisfactory in some parameter ranges.
When $\cdim\gg\rdim$, a sensible improvement would be to replace the matrix
dimension factor $\cdim^\expmatmul$ with one which has the exponent on the
smallest dimension, such as $\rdim^{\expmatmul-1}\cdim$. Similarly, when $\amp
\gg \rdim\deg(\pmat)$, a sensible improvement would be to replace the
polynomial degree factor $\deg(\pmat) + \amp$ with one suggested by the bounds
on $\deg(\popov)$ given in \cref{eqn:bound_deg_popov_global} of
\cref{lem:popov_bounds}.

We achieve both improvements with our second approach, which works in three
steps and is formalised as \cref{algo:knownsupp_popov}.  First, we compute the
$\subTuple{\shifts}{\tuple\piv}$-Popov form of the submatrix
$\matcols{\pmat}{\tuple\piv}$, which can be done efficiently since this
submatrix is square and nonsingular.  Then, we use polynomial matrix division
to obtain the unimodular transformation $\umat \in \polMatSpace[\rdim]$ such
that $\matcols{\pmat}{\pivSupp} = \umat \, \matcols{\popov}{\pivSupp}$.
Lastly, we compute the remaining part of the $\shifts$-Popov form of $\pmat$ as
$\umat^{-1} \matcols{\pmat}{\{1,\ldots,n\}\setminus \tuple\piv}$.  Note that,
even for $\shifts=\unishift$, all entries of $\umat^{-1}$ may have degree in
$\Theta(\rdim \deg(\pmat))$; we avoid handling such large degrees by computing
this product truncated at precision $\var^{\delta}$, where $\delta$ is a
(strict) upper bound on the degree of the $\shifts$-Popov form $\popov$.  For
example, if $\shifts = \unishift$ we can take $\delta = 1+\deg(\pmat)$.

\begin{algobox}
  \algoInfo{KnownSupportPopov}
  \label{algo:knownsupp_popov}

  \dataInfos{Input}{
    \item matrix $\pmat \in \pmatSpace$ with full rank and $\rdim<\cdim$,
    \item shift $\shifts \in \shiftSpace$,
    \item the $\shifts$-pivot support $\tuple\piv = \pivSupp$ of $\pmat$,
    \item bound $\degBd\in\ZZp$ on the degree of the $\shifts$-Popov form of $\pmat$.
  }

  \dataInfo{Default}{
    $\degBd = 1+\min(\sumVec{\rdeg{\pmat}},\sumVec{\cdeg{\pmat'}},\deg(\pmat)+\amp)$, \\
    \phantom{Default: } where $\pmat'$ is $\pmat$ with zero columns removed.
  }

  \dataInfo{Output}{
    the $\shifts$-Popov form of $\pmat$.
  }

  \algoSteps{
    \item $\popov \assign$ zero matrix in $\pmatSpace$
    \item $\matcols{\popov}{\tuple\piv} \assign \algoname{NonsingularPopov}(\matcols{\pmat}{\tuple\piv}, \subTuple{\shifts}{\tuple\piv})$
          \label{step:KSP:PopovPart}
    \item $\umat \assign \matcols{\pmat}{\tuple\piv} \matcols{\popov}{\tuple\piv}^{-1} \in \polMatSpace[\rdim][\rdim]$
          \label{step:KSP:CompU}
    \item $\degBd \assign  \min(\degBd,\ 1+ \max(\rdeg[\subTuple{\shifts}{\tuple\piv}]{\matcols{\popov}{\tuple\piv}}) - \min(\subTuple{\shifts}{(1,\ldots,\cdim)\setminus\tuple\piv})$
          \label{step:KSP:precision}
    \item $\matcols{\popov}{\{1,\ldots,\cdim\}\setminus\tuple\piv} \assign \mat{U}^{-1} \, \matcols{\pmat}{\{1,\ldots,\cdim\}\setminus\tuple\piv} \bmod \var^{\degBd}$
          \label{step:KSP:PopovRest}
    \item \algoword{Return} $\popov$
  }
\end{algobox}

\begin{proposition}
  \label{prop:algo:knownsupp_popov}
  \cref{algo:knownsupp_popov} is correct and uses $\softO{\rdim^{\expmatmul-1}
  \cdim \degBd}$ operations in $\field$, where
  \[
    \degBd = 1+\min(\sumVec{\rdeg{\pmat}},\sumVec{\cdeg{\pmat'}},\deg(\pmat)+\amp),
  \]
  and $\pmat'$ is $\pmat$ with zero columns removed.
\end{proposition}
\begin{proof}
  Let $\mat{Q} \in \pmatSpace$ be the $\shifts$-Popov form of $\pmat$.
  For correctness we prove that $\popov = \mat{Q}$.
  The first part of \cref{lem:popov_bounds} shows that indeed
  $\matcols{\mat{Q}}{\tuple\piv} = \matcols{\popov}{\tuple\piv}$, and that
  $\umat = \matcols{\pmat}{\tuple\piv}\matcols{\popov}{\tuple\piv}^{-1} =
  \matcols{\pmat}{\tuple\piv}\matcols{\mat{Q}}{\tuple\piv}^{-1}$ computed at \cref{step:KSP:CompU}
  is the unimodular matrix such that $\pmat = \umat \mat{Q}$.

  The last item of \cref{lem:popov_bounds} proves that the input default
  value of $\degBd$ is more than $\deg(\mat{Q})$. Besides, by definition of
  $\shifts$-pivots and $\shifts$-Popov form, the column $j$ of $\mat{Q}$ has
  degree at most
  \[
    \max(\rdeg[\subTuple{\shifts}{\tuple\piv}]{\matcols{\mat{Q}}{\tuple\piv}}) - \shift{j}
    = \max(\rdeg[\subTuple{\shifts}{\tuple\piv}]{\matcols{\popov}{\tuple\piv}}) - \shift{j}.
  \]
  It follows that $\degBd >
  \deg(\matcols{\mat{Q}}{\{1,\ldots,\cdim\}\setminus\tuple\piv})$ holds after
  \cref{step:KSP:precision}, and thus the submatrix
  $\matcols{\mat{Q}}{\{1,\ldots,\cdim\}\setminus\tuple\piv}$ is equal to the
  truncated product $\umat^{-1} \,
  \matcols{\pmat}{\{1,\ldots,\cdim\}\setminus\tuple\piv} \bmod \var^{\degBd}$
  computed at \cref{step:KSP:PopovRest}. Hence $\mat{Q} = \popov$.

  Now we explain the cost bound.  \cref{step:KSP:PopovPart} uses
  $\bigO{\rdim^{\expmatmul} \deg(\matcols{\pmat}{\tuple\piv})}$ operations, by
  \cref{thm:nonsing_popov_NeiVu}.  \cref{step:KSP:CompU} has the same cost by
  \cref{lem:PM-QuoRem} below; note that $\matcols{\popov}{\tuple\piv}$ is in
  $\subTuple{\shifts}{\tuple\piv}$-Popov form and thus column reduced.  This is
  within the announced bound since
  \[
    \bigO{\rdim^{\expmatmul} \deg(\matcols{\pmat}{\tuple\piv})} \subseteq
    \bigO{\rdim^{\expmatmul-1}\cdim \deg(\pmat)} 
  \]
  and $\deg(\pmat) \le \degBd$ holds by definition of $\degBd$.

  Finally, \cref{step:KSP:PopovRest} costs
  $\softO{\rdim^{\expmatmul-1}\cdim\degBd}$ operations in $\field$: since
  $\umat(0) \in \matSpace[\rdim]$ is invertible, the truncated inverse of
  $\umat$ is computed by Newton iteration in time $\softO{\rdim^\expmatmul
  \degBd}$; then, the truncated product uses $\softO{\rdim^\expmatmul \lceil
(\cdim-\rdim) / \rdim \rceil \degBd}$ operations.
\end{proof}

At \cref{step:KSP:CompU}, we compute a product of the form $\mat{B}
\mat{A}^{-1}$, knowing that it has polynomial entries and that $\mat{A}$ is
column reduced; in particular, $\deg(\mat{B}\mat{A}^{-1}) \le \deg(\mat{B})$
\cite[Lem.\,3.1]{NeiVu17}.  Then, it is customary to obtain $\mat{B}
\mat{A}^{-1}$ via a Newton iteration on the ``reversed matrices'' (see e.g.
\cite[Chap.\,5]{Sarkar11} and \cite[Chap.\,10]{Zhou12}).

\begin{lemma}
  \label{lem:PM-QuoRem}
  For a column reduced matrix $\mat{A} \in \polMatSpace[\rdim]$ and a matrix
  $\mat{B} \in \polMatSpace[\rdim]$ which is a left multiple of $\mat{A}$, the
  quotient $\mat{B} \mat{A}^{-1}$ can be computed using
  $\softO{\rdim^\expmatmul \deg(\mat{B})}$ operations in $\field$.
\end{lemma}
\begin{proof}
  We follow Steps 1 and 2 of the algorithm \algoname{PM-QuoRem} from
  \cite{NeiVu17}, on input $\mat{A}$, $\mat{B}$, and $d = \deg(\mat{B})+1$;
  hence the requirement $\cdeg{\mat{B}} < \cdeg{\mat{A}} + (d,\ldots,d)$ is
  satisfied.  It is proved in \cite[Prop.\,3.4]{NeiVu17} that these steps
  correctly compute the quotient $\mat{B} \mat{A}^{-1}$; yet we do a different
  cost analysis since the assumptions on parameters in
  \cite[Prop.\,3.4]{NeiVu17} might not be satisfied here.

  Step 1 of \algoname{PM-QuoRem} computes a type of reversals $\matt{A}$ and
  $\matt{B}$
  of the matrices $\mat{A}$ and $\mat{B}$: this uses no arithmetic operation.
  These matrices also have dimensions $\rdim\times\rdim$ and the constant
  coefficient of $\matt{A}$ is invertible because $\mat A$ is column reduced.
  Step 2 computes the truncated product $\matt{B} \matt{A}^{-1} \bmod
  \var^{d+1}$, which can be done via Newton iteration in
  $\softO{\rdim^{\expmatmul} d}$ operations in $\field$.
\end{proof}

Since \cref{algo:knownsupp_popov} works for an arbitrary shift, it allows us in
particular to find the Hermite form of $\pmat$ when its Hermite pivot support
is known.  It turns out that the latter can be computed efficiently via a
column rank profile algorithm from \cite{Zhou12}.

\begin{proof}[Proof of \cref{thm:hermite}]
  Here, the integer $\degBd$ is defined as
  \begin{equation}
    \label{eqn:degBd_hermite}
    \degBd = 1+\min(\sumVec{\rdeg{\pmat}},\sumVec{\cdeg{\pmat'}}),
  \end{equation}
  where $\pmat'$ is $\pmat$ with zero columns removed.

  Let $\shifts[h] = (\cdim\degBd,\ldots,2\degBd,\degBd)$.
  By \cref{lem:popov_bounds}, $\degBd$ is more than the degree of the Hermite
  form of $\pmat$; therefore the $\shifts[h]$-Popov form of $\pmat$ is also its
  Hermite form (see \cref{ssec:hermite_def}).  Thus, up to the knowledge of the
  Hermite pivot support $\pivSupp[\shifts[h]]$ of $\pmat$, we can compute the
  Hermite form of $\pmat$ using $\softO{\rdim^{\expmatmul-1} \cdim \degBd}$
  operations via \cref{algo:knownsupp_popov}.

  As mentioned in \cref{ssec:hermite_def}, $\pivSupp[\shifts[h]]$ is also the
  column rank profile of $\pmat$.  It is shown in \cite[Sec.\,11.2]{Zhou12} how
  to use row basis and kernel basis computations to obtain this rank profile in
  $\softO{\rdim^{\expmatmul-1}\cdim \sigma}$ operations, where $\sigma =
  \lceil\sumVec{\rdeg{\pmat}}/\rdim\rceil$ is roughly the average row degree of
  $\pmat$.  We have $\sigma\le 1+\sumVec{\rdeg{\pmat}}$ by definition, and it
  is easily verified that $\sumVec{\rdeg{\pmat}}/\rdim \le
  \sumVec{\cdeg{\pmat'}}$, hence $\sigma \le \degBd$.
\end{proof}

\begin{acks}
  The authors are grateful to Cl\'ement Pernet for pointing at the notion of
  saturation.
\end{acks}


\end{document}